\documentclass[10pt,twocolumn,prl,nofootinbib]{revtex4-2}
\usepackage{amsmath}
\usepackage{lipsum} % for some text
\usepackage{verbatim}
\usepackage{epsfig}
\usepackage{subfigure}
\usepackage{graphicx}
\usepackage{enumitem}
\usepackage{amsfonts}
\usepackage[final]{showlabels}
\usepackage{enumitem}
\usepackage[figuresright]{rotating}
\usepackage{amssymb}
\usepackage{amsmath}
\usepackage{psfrag}
\usepackage{mathtools}
\usepackage[export]{adjustbox}
\usepackage{bm}% bold math
\usepackage[colorlinks,linkcolor=blue,anchorcolor=blue,citecolor=blue,urlcolor=blue]{hyperref}
\usepackage[version=4]{mhchem}
\usepackage[svgnames]{xcolor}
\newtheorem{proposition}{Proposition}%[section]

\def\be{\begin{equation}} \def\ee{\end{equation}}
\def\bea{\begin{eqnarray}} \def\eea{\end{eqnarray}}

\def\bpm{\begin{pmatrix}} \def\epm{\end{pmatrix}}

\newcommand{\bigO}{\mathcal{O}}

\makeatletter
\newcommand*{\balancecolsandclearpage}{%
	\close@column@grid
	\clearpage
}
\makeatother

\begin{document}
	\title{Redundancy from Subsystem Thermalization}
	\author{Xiangyu Cao} 
	\affiliation{Laboratoire de Physique de l'\'Ecole normale sup\'erieure, ENS, Universit\'e PSL, CNRS, Sorbonne Universit\'e, Universit\'e Paris Cit\'e, F-75005 Paris, France}
	\author{Zohar Nussinov}
	
	\affiliation{Physics Department, Washington University in St. Louis, St. Louis, MO 63141, USA}
	\affiliation{LPTMC, CNRS-UMR 7600, Sorbonne Universit\'e, 4 Place Jussieu, 75252 Paris cedex 05, France}

	\date{\today}
	\begin{abstract}
		In the theory of decoherence, redundancy is the correlation between a quantum system and  fractions of the environment. It underlies the emergence of classical behavior. We show that redundancy can persist despite thermalizing dynamics in the environment. This follows an initial broadcasting interaction that changes the density of a conserved quantity. The mutual information between the system and a fraction of the environment is estimated using the large deviation principle governing subsystem thermalization.
	\end{abstract}
	\maketitle
	
	How the classical realm emerges from the laws of quantum mechanics seems to be a timeless question. In the past decades, a fruitful approach has been the decoherence idea~\cite{zeh-71,zeh-73,zurek-81,zurek-deco-review,zeh-review,SCHLOSSHAUER20191}, and its modern refinements, such as Quantum Darwinism~\cite{ollivier-poulin,zurek-QD,zurek-review,Korbicz-rev}. At the heart of this approach lies the notion of \textit{redundancy}, defined as the correlation between a quantum system with multiple small fractions of its environment. Quantitatively, redundancy can be characterized by the independence of the system-fraction mutual information on the fraction's size, often literally captured (as in  Fig. \ref{fig:example} for the example to be discussed in the current work) by a  ``redundancy plateau.'' 
	Conceptually, redundancy can be viewed as underlying the emergence of an ``objective fact,'' in the sense that some information about the system can be attested by multiple observers. Implications on conceptual issues of quantum mechanics, such as the measurement problems, have been amply discussed~\cite{Tomaz-philo,zurek-philo18}.
	
	What kind of quantum dynamics produces redundancy? This question is timely as direct probes of redundancy have become possible in quantum simulation platforms~\cite{unden19-darwin-exp,mondani-darwism}. It is also theoretically nontrivial, as quantum dynamics generically scrambles information, instead of broadcasting it~\cite{deutsch,srednicki,rigol-review}. One might expect some fine-tuning is necessary to maintain redundancy. Indeed, many examples exhibiting redundancy consist of the system broadcasting to non-interacting environment parts~\cite{QD-QBM,girolami,campbell,collision-model,Ryan-onion}. Riedel, Zurek, and Zwolak~\cite{riedel2012} showed an example where redundancy, initially established by the system-environment interaction, fades away due to the environment's intrinsic dynamics. Deffner and collaborators~\cite{duruisseau,deffner} examined a large class of Hamiltonian systems, and concluded that redundancy limits drastically intra-environment interaction. One of us~\cite{fertecaoprl,fertecaopra,ferte2025decoherent} found robust redundancy, yet with an expanding environment incorporating more and more qubits. 
	
	In this Letter, we point out a generic and robust mechanism of redundancy with an interacting environment. We assume the environment to be governed by a non-integrable many-body Hamiltonian, so that its intrinsic dynamics leads to thermalization. We shall show redundancy generically arises in the long time and large-environment limit, after a ``broadcasting'' interaction with the system.  
	Such an interaction yields a Shr\"odinger-cat like superposition of states with different energy densities \cite{nussinov20}. Thermalization converts these states into locally distinguishable macrostates which are redundantly imprinted in the environment. Our result does not contradict Ref.~\cite{riedel2012} which analyzed a fine-tuned exceptional case where the states shared the same  energy density and thus  thermalized to the same macro-state, resulting in an ``encoding" behavior. As a technical contribution, we quantitatively relate redundancy to a large-deviation form of subsystem thermalization~\cite{dymarsky-subsystem-eth,rigol-review,touchette-ldp}.

	\begin{figure}
		\centering
		\includegraphics[width=\columnwidth]{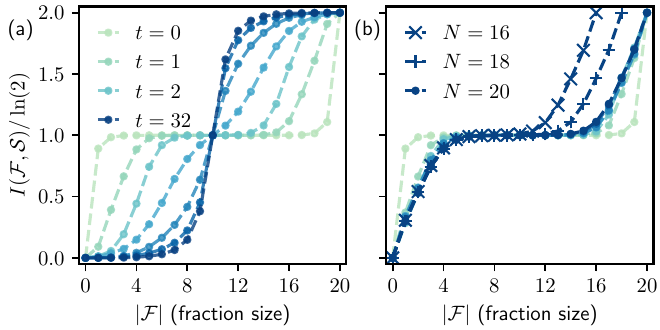}
		\caption{
			Mutual information $I(\mathcal{F}, \mathcal{S})$ between the system and a 
			fraction $\mathcal{F}$ of the environment, as a function of its size $|\mathcal{F}|$, 
			after a broadcasting interaction \eqref{eq:H-int} (with $t_0 = \pi / 4\lambda \times 0.75$), 
			followed by intrinsic environment evolution under \eqref{eq:HE-ex} for 
			$t = 1, 2, 4, \dots, 32$ (light to dark). 
			(a) The initial redundancy plateau crosses over to encoding behavior due to thermalization. 
			(b) Upon replacing $Z_j$ by $Y_j$ in \eqref{eq:H-int}, the redundancy plateau persists at late times and the approach to it does not depend on the environment size $N$.  }\label{fig:example}
	\end{figure}
	\textit{Basic example.}-- We illustrate our point in a concrete example. We consider the system to be a qubit (spin-$1/2$), and the environment to be made of $N \gg 1$ qubits. A well-known way to generate redundancy is the following. We initialize the system-environment pair in the product state 
	\begin{equation}
		| \Psi(0) \rangle = | + \rangle_{\mathcal{S}} \otimes | + \rangle^{\otimes N}_{\mathcal{E}}, 
	\end{equation}
	where $ | + \rangle = (| 0 \rangle + | 1 \rangle) / \sqrt{2}$.  Then we let them interact with the broadcasting Hamiltonian~\cite{blume-kohout-2005,riedel2012,duruisseau,girolami,nussinov20},
	\begin{equation} \label{eq:H-int}
		H_{\text{int}} =  - \lambda Z_\mathcal{S}  \sum_{j=1}^N Z_j,
	\end{equation}
	where $Z_\mathcal{S}$ is the Pauli $Z$ operator acting on $\mathcal{S}$, and $Z_j$ the Pauli $Z$ operator acting on the $j$th environment qubit. After a time $t_0$ such that $\lambda t_0 = \pi / 4$, the system and its environment are described by the state 
	\begin{equation}  \label{eq:Psi-init}
		| \Psi \rangle = \sum_{a=0,1} c_a | a \rangle_{\mathcal{S}}   | \Phi_a \rangle_{\mathcal{E}}.
	\end{equation}
	Here, $c_a = 1 / \sqrt{2}$ and $| \Phi_a \rangle = | \varphi_{a} \rangle^{\otimes N}$.
	The states $|\varphi_0\rangle$ and $|\varphi_1\rangle$ are, respectively, $|+_y\rangle$ and $|-_y\rangle$ (the latter being polarized in opposite $y$ directions). The  state $| \Psi \rangle$ exhibits perfect redundancy. Each environment spin is perfectly correlated with ${\mathcal{S}}$. A single spin suffices to determine the system state.
	The redundancy is quantitatively captured by the mutual information between an environment fraction $\mathcal{F}$ of $n$ qubits and $\mathcal{S}$,  
	\begin{equation*}
		I (\mathcal{F}, \mathcal{S}) = H(\mathcal{F}) + H( \mathcal{S}) - H( \mathcal{F}  \mathcal{S}) 
	\end{equation*}
	where $H(\mathcal{X}) = \mathrm{tr}[-\rho_{\mathcal{X}}\ln \rho_{\mathcal{X}}]$ is the von Neumann entropy.  It is straightforward to see that $I(\mathcal{F}, \mathcal{S})  = \ln 2$ for any relative fraction size $f := n / N \in (0, 1)$: We have a redundancy plateau. In fact, there is a redundancy plateau in the $ N \to \infty $ limit for any $ \lambda t_0 \ne k \pi / 2$. 
	
	Next, we turn off the interaction and let the environment evolve under the chaotic Ising Hamiltonian
	\begin{equation} \label{eq:HE-ex}
		H_{\mathcal{E}} = - \sum_{j=1}^N \left(  J Z_j Z_{j+1} + h_X X_j + h_Z Z_j \right),
	\end{equation}
	where $J = 1, h_Z = 1.205, h_X = 0.945$, while $\mathcal{S}$ does not evolve. We then compute the mutual information between an interval $\mathcal{F}$ and $\mathcal{S}$; by translation invariance this only depends on the size $ n = |\mathcal{F}|$, not on its position. In Figure~\ref{fig:example}-(a), we observe that the redundancy plateau gradually disappears and is replaced by an encoding ``jump'': $I(\mathcal{F}, \mathcal{S}) \to 0$ for all $n < N/2$ while $I \to 2 \ln 2$ as soon as $n > N/2$. Namely, no correlation persists between any small fraction and the system. We have thus reproduced the redundancy-encoding crossover observed in Ref.~\cite{riedel2012} in a similar setting. 
	
	However, if we repeat the numerical test with $Z_j$ replaced by $Y_j$ in the broadcasting interaction \eqref{eq:H-int}, the initially established redundancy persists despite the chaotic evolution, see Fig.~\ref{fig:example}-(b).
	After an initial transient, $I(\mathcal{F}, \mathcal{S})$ becomes time-independent, and tends to $\ln 2$ exponentially with increasing $n$.
	The mutual information $I(\mathcal{F}, \mathcal{S})$ is also independent of the total environment size $N$. So, in the $N \to \infty$ limit, there is a perfect redundancy plateau, $ I(\mathcal{F}, \mathcal{S}) \approx \ln 2$ for any relative size $ 0 < n/N  < 1$. 
	
	The distinct fates of redundancy stem from the environment's energy density in the different branches of the state~\eqref{eq:Psi-init}:
	\begin{equation} \label{eq:e-average}
		\epsilon_a =  \langle   \Phi_{a}  | H_{\mathcal{E}}	|  \Phi_{a} \rangle / N. 
	\end{equation}
	It is straightforward to check that $\epsilon_0 = \epsilon_1$ with the interaction \eqref{eq:H-int} unmodified. By replacing $Z_j$ by $Y_j$, the energy densities become distinct $\epsilon_0 \ne \epsilon_1$. The former situation is exceptional and the latter {generic}: if we replace $Z_j$ by $x X_j + y Y_j + z Z_j$ where $(x, y, z)$ is drawn randomly on the unit sphere, then $\epsilon_0 \ne \epsilon_1$ with probability one. Next, we shall focus on the generic case and explain the above observation analytically and in more general setup. 
	
	\textit{Redundancy and subsystem thermalization.}--
	As a general setup, we now  consider a system with a $d$-dimensional Hilbert space and computational basis $|a\rangle$, $a = 0, \dots, d-1$. The environment has $N \gg 1$ degrees of freedom and evolves under a non-integrable local Hamiltonian $H_{\mathcal{E}}$, with energy as the only conserved quantity.
	Assume that initial system-environment interaction established a state,
	\begin{equation} 
		| \Psi \rangle =  \sum_{a=0}^{d-1} c_a | a \rangle_{\mathcal{S}} | \Phi_a \rangle_{\mathcal{E}}
	\end{equation}
	where $|c_a|^2 > 0$ for all $a$,  $\sum_i |c_a|^2 = 1$ and $| \Phi_a \rangle$ are short-range correlated states such that the energy densities $
	\epsilon_a$, as defined in \eqref{eq:e-average}, are pairwise distinct. Time evolution of the environment gives
	\begin{equation} \label{eq:phit}
		| \Psi(t) \rangle =  \sum_{a=1}^d c_a | a \rangle_{\mathcal{S}}  | \Phi_a(t) \rangle_{\mathcal{E}} ,~   | \Phi_a(t) \rangle  = e^{-i H_{\mathcal{E}} t} | \Phi_a \rangle.
	\end{equation}  
	We wish to estimate the mutual information $I(\mathcal{F}, \mathcal{S})$ for the state $ | \Psi(t) \rangle$, and for a fraction with size $n = | \mathcal{F} |$, under a subsystem thermalization hypothesis~\cite{dymarsky-subsystem-eth,rigol-review,touchette-ldp}: 
	
	\noindent	\textbf{Assumption}.
	For any $a$, the probability distribution of the energy of the fraction, 
	with respect to the state $|\Phi_a(t)\rangle$, satisfies a large deviation principle with rate function $f_a(\epsilon) \geq 0$,
	\begin{equation} \label{eq:large-deviation}
		p_a(\epsilon) := \langle \Phi_a(t) | \delta(H_{\mathcal{F}} - \epsilon n ) | \Phi_a(t) \rangle  \sim e^{- n f_a(\epsilon)}.
	\end{equation}
	Here, $ H_{\mathcal{F}} $ is the Hamiltonian restricted to $\mathcal{F}$, and $X \sim e^{y n}$ means $ n^{-1} \ln X \to y$ as $n \to \infty$. Strictly speaking, the Dirac delta  function in \eqref{eq:large-deviation} is  smeared, so that $p_a(\epsilon)$ is smooth. The $n \to \infty$ limit is independent of the 
	smearing width, provided it is held fixed.  Equivalently, 
	the large deviation principle can be  formulated in terms of the generating function~\cite{ellis,gartner}:
	\begin{equation} \label{eq:lambda}
		\langle \Phi_a(t) | e^{k H_{\mathcal{F}} } | \Phi_a(t) \rangle  \sim e^{n \lambda_a(k)}, 
	\end{equation}
	where  $\lambda_a$ is the Legendre transform of $f_a$, $\lambda_a(k) = \max_{\epsilon} (k \epsilon - f_a(\epsilon))$. 
	
	To motivate the assumption, note that the energy distribution of the whole environment, that is  \eqref{eq:large-deviation} with $\mathcal{F} = \mathcal{E}$, is conserved 
	and fixed by the initial condition $|\Phi_a\rangle$.   Not only the mean energy, but also all cumulants of the energy distribution are conserved. As long as $  | \Phi_a (t) \rangle$ has sufficiently short-range correlation, the latter cumulants are extensive quantities. Our assumption amounts to saying that all cumulant densities are the same in large enough subsystems; hence, it is a natural generalization of the equipartition law. Note that $\lambda_a(k)$ in \eqref{eq:lambda} is the cumulant generating function density.
	
	Combining the above assumptions and basic quantum information theory, we may show 
	\begin{proposition}   \label{prop:main}
		For any $\alpha < \alpha_*$ where
		\begin{align} \label{eq:alpha}
			&	 \alpha_* =  \min_{a \ne b} \min_{\epsilon} \max(f_a(\epsilon), f_b (\epsilon)), 
		\end{align} 
		there is $C > 0$ such that for any fraction $|\mathcal{F}|$ of size $n$, 
		\begin{equation}  \label{eq:plateaumain}
			H(\mathcal{S})  - C e^{-  \alpha n }  \le	I(\mathcal{F}, \mathcal{S})  \le  H(\mathcal{S})  +  C e^{-  \alpha (N - n) }. 
		\end{equation}
	\end{proposition}
	The proof (see End Matter) proceeds as  follows. The quantum mutual information $I(\mathcal{F}, \mathcal{S}) $ is bounded from below by the classical mutual information between the subsystem energy in $\mathcal{F}$ and $a$, the outcome of a projective measurement of $\mathcal{S}$ in the computational basis. These variables are strongly correlated, as the total energy is typically close to $\epsilon_a n$. This gives a large classical mutual information close to the Shannon entropy of the distribution $\{|c_a|^2\}_{a=0}^{d-1}$, which we may show is close to  $H(\mathcal{S})$, 
	\begin{equation} \label{eq:HS-main}
		H(\mathcal{S})  =  - \sum_a |c_a|^2  \ln |c_a|^2 + \bigO( e^{- \gamma N}), \ \gamma > 0. 
	\end{equation} 
	In the above example, $c_0 = c_1 = 1/\sqrt{2}$, which gives $ H(\mathcal{S}) \approx \ln 2$, as observed in Fig.~\ref{fig:example}-(b).  
	
	\begin{figure}
		\centering
		\includegraphics[width=\columnwidth]{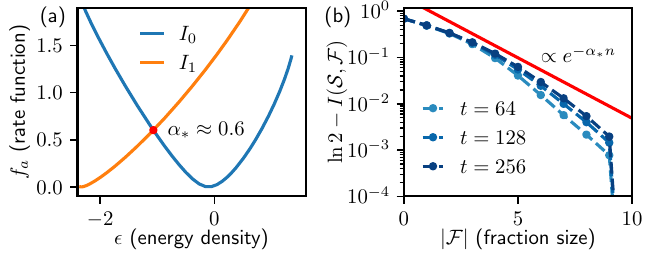}
		\caption{(a) Conserved energy density large deviation rate functions \eqref{eq:large-deviation} in the example of Fig.~\ref{fig:example}-(b). The intersection point gives $\alpha_*$ in \eqref{eq:alpha}.  (b) Approach the plateau of $I(\mathcal{F}, \mathcal{S})$ [same condition as Fig.~\ref{fig:example}-(b), $N = 20$, larger $t$] compared to the bound \eqref{eq:plateaumain}.} \label{fig:estimate}
	\end{figure}	
	Eq.~\eqref{eq:plateaumain} implies that the mutual information approaches the plateau value at least exponentially fast in the fraction size $n/N$. The exponent $\alpha_*$ depends on the large deviation rate functions $f_a(\epsilon)$. In fact, assuming convexity of $f_a$, for any $a \ne b$, the minimum $\min_{\epsilon} \max (f_a(\epsilon), f_b(\epsilon))$ is obtained at the unique locus $\epsilon_*$ between $\epsilon_a$ and $\epsilon_b$ where $f_a(\epsilon_*)= f_b(\epsilon_*) = \alpha_*$. As an illustration, in Fig.~\ref{fig:estimate} we plot the numerical estimated rate functions $f_a$ and determine $\epsilon_*$ graphically.
	We then verify that $|\ln 2 - I(\mathcal{S}, \mathcal{F})|$ approaches zero faster 
	than $e^{-\alpha_* n}$, confirming that the mutual information reaches its long-time 
	plateau exponentially fast in $n$. This bound becomes tight at late times. Note that the correction to the plateau is dominated by an atypical energy density 
	$\epsilon_*$, which differs from the typical value $\epsilon_a$ for all $a$. This is the reason why we needed the large deviation form of subsystem thermalization hypothesis.  
	
	\textit{Redundancy number.}-- An application of Proposition \ref{prop:main} is to estimate the redundancy number $R_\delta$, defined as the number of fractions whose mutual information with the system is equal the plateau value, up to a tolerance $\delta$,
	\begin{equation}
		\label{red_number}
		R_{\delta} := \max_{	I(\mathcal{F}, \mathcal{S}) \ge	H(\mathcal{S}) - \delta } \left(  | \mathcal{E} | / | \mathcal{F} |  \right) 
	\end{equation}
	From \eqref{eq:plateaumain}, we find that the redundancy number is extensive, that is, scales linearly with the environment size as the latter increases, 
	\begin{equation}
		R_{\delta} \gtrsim \frac{ \alpha_* }{\ln (1/\delta) }  | \mathcal{E} |  . 
	\end{equation}
	Observe that the proportionality constant in this inequality decreases very slowly 
	as $\delta \to 0$: there exists a large number of fractions with a high-quality 
	redundant copy. 
	
	\begin{figure}
		\centering\includegraphics[width=\columnwidth]{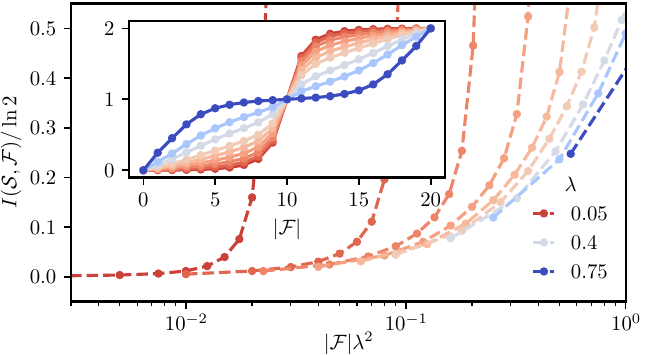}
		\caption{System-fraction mutual information $I(\mathcal{F}, \mathcal{S})$ in an interpolation between Fig.~\ref{fig:example}-(a) and (b). The parameters are the same as in Fig.~\ref{fig:example}, except that $Z_j \to Z_j \cos(\lambda \pi / 2) + Y_j \sin(\lambda \pi / 2)$ in~\eqref{eq:H-int}. In the main plot, we show the data collapse as a function of $\lambda^2  |\mathcal{F}|$. In the inset, we show the raw data.}\label{fig:crossover}
	\end{figure}
	
	\textit{Degeneracy}.-- 
	Proposition~\ref{prop:main} is useful for interpreting finite-size numerical data, 
	particularly when nearly degenerate energy densities obscure the redundancy plateau.  
	When two energy densities merge, $\epsilon_{0, 1} = \bar{\epsilon} \pm \delta \epsilon$ with $\delta \epsilon \to 0$, 
	and assuming that $f''_{a}(\epsilon_a) = c + \bigO(\delta \epsilon)$ for both $a = 0, 1$,  \eqref{eq:alpha} implies  $\alpha_* \le \frac{c}2 (\delta \epsilon)^2  + \bigO\left( (\delta \epsilon)^3 \right)$. Hence, for $c > 0$, a redundancy plateau appears only for
	\begin{equation}
		|\mathcal{F}| \gtrsim n_* :=  1/ (\delta \epsilon)^2  . \label{eq:scaling-degen} 
	\end{equation}
	To verify this scaling, we interpolate between the two examples in Fig.~\ref{fig:example}  by replacing $Z_j$ with $Z_j \cos(\lambda\pi/2) + Y_j \sin(\lambda\pi/2)$, $0 < \lambda < 1$, 
	in \eqref{eq:H-int}. Note that $\delta\epsilon \sim \lambda$ as $\lambda \to 0$. 
	As $\lambda$ varies, Fig.~\ref{fig:crossover} shows that $I(\mathcal{S}, \mathcal{F})$ 
	can exhibit a plateau, a ramp, or encoding behavior.
	Nevertheless, plotting $I(\mathcal{S}, \mathcal{F})$ as a function of 
	$\lambda^2 |\mathcal{F}| \sim (\delta\epsilon)^2 |\mathcal{F}|$ yields an excellent 
	data collapse for $|\mathcal{F}| \ll |\mathcal{E}|/2$, corroborating the scaling 
	law~\eqref{eq:scaling-degen}. This suggests that ramp-like behavior~\cite{deffner,duruisseau} 
	may signal a slow crossover to redundancy not yet visible at finite size. In End Matter we discuss a more complex situation where some (but not all) $\epsilon_a$ are degenerate.

	\textit{All-to-all environment.}-- We assumed that $H_{\mathcal{E}}$ is a local Hamiltonian. This choice differs from many works on quantum Darwinism which considered an  environment with random all-to-all $q$-local interaction. An example with $q = 2$ is 
	\begin{equation}
		H_{\mathcal{E}} =   \frac1{\sqrt{N}} \sum_{j < k}  \left( J^x_{jk} X_{j} X_{k} + J^z_{jk} Z_{j} Z_{k} \right),
	\end{equation}
	where $J^{\alpha}_{jk}$ are independent standard Gaussian random couplings. 
	(The Sachdev-Ye-Kitaev~\cite{sachdevye,kitaev} model is a $q=4$-local example with 
	fermions.) Observing redundancy in such environments is challenging for two reasons. First, 
	altering the energy density via a broadcasting interaction is not straightforward: 
	a spin coherent state has $\epsilon = 0$ in the $N \to \infty$ limit~\cite{scaffidi}. 
	Second, energy is stored non-locally. In systems with all-to-all $q$-local interactions, the energy of a fraction $\mathcal{F}$ of 
	size $n$, defined as the sum of intra-fraction couplings, scales as~\cite{zheng-liu-chen-syk,gu-huang-syk}
	\begin{equation}
		E_{\mathcal{F}} \sim n(n/N)^{q-1}.
	\end{equation}
	The factor $(n / N)^{q-1}$ is absent in short-range systems. 
	Thus, $E_{\mathcal{F}} \gtrsim 1$, with energy differences also $O(1)$, if and 
	only if $n \gtrsim n_* = N^{(q-1)/q}$. The redundancy number $R_{\delta}$ [defined in \eqref{red_number}] is subextensive, $R_{\delta} \sim N/n_* = N^{1/q}$, and the redundancy plateau is visible only for 
	$n_* \ll N/2$, which requires large $N$.
	
	\begin{figure}
		\centering
		\includegraphics[width=\columnwidth]{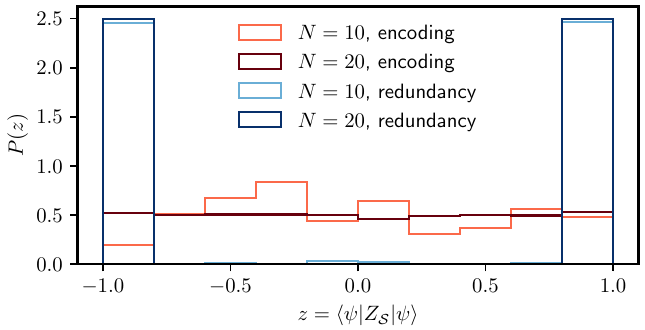}
		\caption{
			Distribution of expectation value $ \left< \psi | Z_{\mathcal{S}} | \psi \right>$ where $\psi$ is drawn from the projective ensemble obtained from measuring the environment in the computational basis in the final states in Fig.~\ref{fig:example} with $N = 20$ and $N = 10$ (same other parameters). In the encoding case 
			[Fig.~\ref{fig:example}-(a)], the distribution is approximately uniform; in the 
			redundancy case [Fig.~\ref{fig:example}-(b)], it is peaked at $\pm 1$.
		} \label{fig:proj}
	\end{figure}
	\textit{Pointer state and projective ensemble.}-- 
	We have focused on the mutual information $I(\mathcal{F}, \mathcal{S})$, which 
	quantifies the total system-environment correlation as a single number but is 
	nontrivial to estimate. A more practical probe of the correlation is the ensemble 
	of posteriori states of the system conditioned on projective measurements on the 
	environment, known as the projective ensemble~\cite{deep-therm,deep-therm-theory,deep-therm-ho}. A partial or weak measurement instead gives rise to the pointer state 
	ensemble~\cite{ferte2025decoherent,fertecaopra}.
	
	For the examples in Fig.~\ref{fig:example}, we construct the projective ensemble 
	via computational basis measurements. With redundancy, the ensemble tends to a sum 
	of delta peaks at the classical pointer states $|a\rangle\langle a|$. In the encoding 
	case, it tends to a uniform distribution on the Bloch sphere, as in deep 
	thermalization~\cite{deep-therm,deep-therm-theory}. The two cases are sharply 
	distinguished by the marginal distribution of $\langle\psi|Z_{\mathcal{S}}|\psi\rangle$, 
	where $|\psi\rangle$ is drawn from the ensemble: it is peaked at $\{-1,1\}$ in the 
	former and uniform on $[-1,1]$ in the latter. This distinction is already dramatic 
	in small systems, where the post-selection problem can be brute-forced.
	
	\textit{Discussion.}-- 
	We showed that redundancy robustly arises from two ingredients: a local conserved 
	quantity of the environment, and a ``broadcasting interaction'' with the system that 
	alters the environment macroscopically. Our results generalize directly to other 
	conserved quantities beyond energy. Thus, redundancy is more robust than suggested 
	by the existing literature. Yet the broadcasting interaction is a radical ingredient: 
	it is non-local and involves a hierarchy of scales, where a single system degree of 
	freedom acts as a strong point source affecting a large number of environment degrees 
	of freedom. Inflationary models~\cite{fertecaoprl,fertecaopra,ferte2025decoherent} 
	realize this hierarchy more gradually. This unusual situation is known to generate macroscopic fluctuations~\cite{nussinov20}: a single distant drive must collectively couple to many local degrees of freedom in order to alter an intensive quantity such as the energy density at a finite rate. In the current context, multiple components of the environment (including those associated with  macroscopic observers) are not inert. Instead, they collectively interact with the single monitored system which consequently imprints its state redundantly across them. Such macroscopic  fluctuations can 
	also arise from ``kinetically constrained'' \textit{local} interactions~\cite{fagotti1,fagotti2}, which we expect to be another mechanism of redundancy. An impurity may also induce redundancy by interacting locally with a \textit{zero}-temperature bath~\cite{milazzo-backflow}. It appears that redundancy, which plausibly underlies quantum measurements, necessitates a low-entropy arrangement of some kind. Formulating this surmise more precisely is an important open problem. Ref.~\cite{LLatune2025thermodynamically} took an interesting step in this direction using quantum thermodynamics~\cite{Campbell_2026_roadmap}.
	
	\begin{acknowledgments}
		We thank the  Laboratoire de Physique Th\'eorique et de la Mati\`ere Condens\'ee, where this work was initiated, for their hospitality. 
	\end{acknowledgments}
	
	\bibliography{ref}
	
	\pagebreak
	\section{End matter}
	\noindent 
	\textit{Proof of Proposition 1}.
	To start, we compute the reduced density matrix of $\mathcal{SF}$:
	\begin{align} 
		&\rho :=	\rho_{\mathcal{SF}} = \sum_a |c_a|^2  | a \rangle \langle a | \rho_{aa} + \sum_{a \ne b} c_a c_b^* 
		|a \rangle \langle b | \rho_{ab}  \\
		& \text{where }	\rho_{ab} := \mathrm{Tr}_{\mathcal{E} \setminus \mathcal{F}} \left[| \Phi_a(t) \rangle \langle \Phi_b(t) | \right].   \label{eq:rhoij}
	\end{align}
	We now consider applying the dephasing channel 
	$\mathbf{D}_{\mathcal{S}}$ to $\mathcal{S}$ in the computational basis, and the dephasing 
	channel $\mathbf{D}_{\mathcal{F}}$ to $\mathcal{F}$ in the eigenbasis of $H_{\mathcal{F}}$ 
	(the Hamiltonian restricted to $\mathcal{F}$).  As a result, we obtain the density matrix 
	\begin{align}
		\rho' :=& \mathbf{D}_{\mathcal{S}} \mathbf{D}_{\mathcal{F}}[\rho] \nonumber \\
		= & \sum_{a, \mu} |c_a|^2  |a, E_{\mu} \rangle \langle a, E_\mu |   
		\langle E_{\mu} | \rho_{aa} | E_{\mu} \rangle
	\end{align}
	where $| E_{\mu} \rangle$, $\mu = 1, \dots, D$, are the energy eigenstates of $H_{\mathcal{F}}$, 
	and $D$ is the Hilbert space dimension of $\mathcal{F}$. As is well known~\cite{nielsen_chuang_2010}, dephasing cannot decrease the entropy of $\mathcal{S}$,
	\begin{equation} \label{eq:boundH}
		H(\mathcal{S}) \le  H'(\mathcal{S}) =   - \sum_{a} |c_a|^2 \ln |c_a|^2 ,
	\end{equation}
	and the mutual information cannot increase,
	\begin{equation} \label{eq:bound}
		I(\mathcal{F}, \mathcal{S})  \ge 	I'(\mathcal{F}, \mathcal{S}). 
	\end{equation}
	Here and in what follows, primed quantities denote computations performed with  
	$\rho'$.
	
	The density matrix $\rho'$ is diagonal in the basis $| a, E_\mu \rangle$, 
	$a = 0, \dots, d-1$, $\mu = 1, \dots, D$. It describes two classical random 
	variables $a$ and $\epsilon$ with the following joint law: We draw $a \in \{0, \dots, d-1\}$ 
	with probability $|c_a|^2$, and then, conditioned on $a$, draw $\epsilon = E_\mu / n$ 
	with probability $\langle E_{\mu} | \rho_{aa} | E_{\mu} \rangle$. Therefore, the joint 
	distribution is
	\begin{align} 
		p(a, \epsilon) := |c_a|^2 p_a (\epsilon) 
	\end{align}   
	where $p_a(\epsilon) = \langle \Phi_a(t)| \delta(H_\mathcal{F} - n \epsilon) 
	| \Phi_a(t) \rangle = \mathrm{Tr}[\rho_{aa} \delta(H_\mathcal{F} - n \epsilon)]$ 
	is the energy density distribution in $\mathcal{F}$ on $n$ qubits of the environment, as defined in \eqref{eq:large-deviation}. In the large $n$ limit, the sum over energy eigenstates $\mu$ can be approximated by an integral over the rescaled energy $\epsilon$, with $p_a(\epsilon)$ playing 
	the role of the continuous energy density, as made precise by the large deviation 
	assumption \eqref{eq:large-deviation}. 
	The mutual information $I'(\mathcal{F}, \mathcal{S})$ is the classical Shannon mutual 
	information between the random variables $\epsilon$ and $a$.
	
	Consequently, we have
	\begin{align}
		&\delta := - \sum_a |c_a|^2 \ln |c_a|^2  - I'(\mathcal{F}, \mathcal{S})  \nonumber \\ 
		=&   \sum_a  |c_a|^2 \int \mathrm{d} \epsilon \, p_a(\epsilon)   \ln \left(1 +  \sum_{b \ne a} \frac{|c_b|^2 p_b(\epsilon) }{ |c_a|^2 p_a(\epsilon)} \right). \label{eq:mutualclassical}
	\end{align}	
	Next we show that $\delta$ is exponentially small in $n$. We use the assumption 
	\eqref{eq:large-deviation}, $p_a(\epsilon) \sim e^{-n f_a(\epsilon)}$, to estimate 
	the logarithm in \eqref{eq:mutualclassical}:
	\begin{align}
		& \ln \left(1 +  \sum_{b \ne a} \frac{|c_b|^2 p_b(\epsilon) }{ |c_a|^2 p_a( \epsilon)} \right)  \nonumber \\ 
		\sim & \begin{cases}
			n  & \exists b,  f_b(\epsilon) \le f_a(\epsilon),   \\ 
			e^{- n ( \min_{b \ne a} f_b(\epsilon) - f_a(\epsilon)) }   & \forall b,  f_b(\epsilon) > f_a(\epsilon).
		\end{cases} \label{eq:log-est}
	\end{align}
	Indeed, in the first case, the argument of the logarithm is $\sim e^{\lambda n}$ for some 
	$\lambda \ge 0$, so that the logarithm itself is $O(n)$; in the second case, the argument is 
	$1 + O(e^{-\lambda n})$ for $\lambda > 0$, so the log is $\sim e^{-\lambda n}$, 
	which is exponentially small. This leads the following estimate for the integrand in \eqref{eq:mutualclassical},
	\begin{align} 
		\label{eq:estimate}
		& \sum_{a} |c_a|^2 p_a(\epsilon n) \ln \left(1 +  \sum_{b \ne a} 
		\frac{|c_b|^2 p_b(\epsilon n) }{ |c_a|^2 p_a(\epsilon n)} \right) \nonumber  \\
		\sim & \exp\left[- n \min_a \max (f_a(\epsilon) , \min_{b\ne a} f_b(\epsilon)) \right] 
		\nonumber   \\
		\sim & \exp\left[- n \min_{a \ne b} \max(f_a(\epsilon) , f_b(\epsilon)) \right]. 
	\end{align}
	This implies that $\delta \sim e^{-\alpha_* n}$, with $\alpha_*$ defined in 
	\eqref{eq:alpha}. It follows that for any $\alpha < \alpha_*$, there exists $C > 0$ 
	such that
	$$
	I'(\mathcal{F}, \mathcal{S}) \ge H'(\mathcal{S}) - C e^{-\alpha n}.
	$$
	Fusing this bound with Eqs. (\ref{eq:boundH}, \ref{eq:bound}) leads to \begin{align}
		I(\mathcal{F}, \mathcal{S}) \ge I'(\mathcal{F}, \mathcal{S})  \ge H'(\mathcal{S}) 
		- C e^{-\alpha n} \ge H(\mathcal{S}) - C e^{-\alpha n}, \label{eq:final}
	\end{align}
	which gives the lower bound for $I(\mathcal{F}, \mathcal{S})$ in \eqref{eq:plateaumain}.
	
	To demonstrate the upper bound of   \eqref{eq:plateaumain} on $I(\mathcal{F}, \mathcal{S})$, we turn to the following argument. Since the system-environment hybrid 
	$\mathcal{ES}$ is in a pure state (so that $H(\mathcal{ES}) = 0$ and 
	$H(\mathcal{E}) = H(\mathcal{S})$), we have the identity 
	$I(\mathcal{F}, \mathcal{S}) + I(\mathcal{E} \setminus \mathcal{F}, \mathcal{S}) = 2H(\mathcal{S})$, 
	which combined with \eqref{eq:final} implies 
	$I(\mathcal{E} \setminus \mathcal{F}', \mathcal{S}) \le H(\mathcal{S}) + C e^{-\alpha n'}$ 
	for any fraction $\mathcal{F}'$ of size $n'$. Taking $\mathcal{F}'$  
	such that 
	$\mathcal{F} \subset \mathcal{E} \setminus \mathcal{F}'$, so that 
	$n' = N - n$, and using the fact that mutual information increases when enlarging 
	a subsystem, we arrive at the upper bound:
	\begin{equation}\label{eq:final1}
		I(\mathcal{F}, \mathcal{S}) \le H(\mathcal{S}) + C e^{-\alpha(N-n)}.
	\end{equation}
	This completes the proof. 
	
	We remark that, combining \eqref{eq:final}, \eqref{eq:final1} gives us 
	\begin{equation}
		H'(\mathcal{S}) - H(\mathcal{S})  \le C (e^{-\alpha (N-n) }   + e^{-\alpha n})
	\end{equation}
	for any $1 \ll n \ll N$. Since the left hand side is independent of $n = | \mathcal{F}|$, it must decay exponentially in $N$, $ H'(\mathcal{S}) - H(\mathcal{S}) = O(e^{- \gamma N})$ for some $\gamma$.

	\noindent\textit{Partial degeneracy.} When some (but not all) energy densities coincide, we expect the broadcast of part of the information carried by $\mathcal{S}$, and a redundancy plateau at a different value, which is given by the Shannon entropy of the distribution $\{  |c_a|^2 \}_{a=0}^{d-1}$ coarse-grained by merging entries of the same energy, 
	\begin{equation} \label{eq:plateau-partial}
		I(\mathcal{F}, \mathcal{S}) \approx   - \sum_{\epsilon \in \{ \epsilon_a \}}  p_{\epsilon} \ln p_{\epsilon}, p_{\epsilon} = \sum_{\epsilon_a = \epsilon}  |c_a|^2 .
	\end{equation}
	Indeed, observing the environment fraction enables distinguishing branches of the wave  function with distinct energy but cannot resolve those with the same energy. Eq.~\eqref{eq:plateau-partial} is expected to hold when $1 \ll |\mathcal{F}| < |\mathcal{E}|/2$, with an ``encoding jump'' occurring as $|\mathcal{F}|$ approaches this upper bound.
	
	To illustrate the above, we consider an example of the general setup with $d = 3$, $| \Phi_{0,2} \rangle = ((| 0 \rangle \pm i | 1 \rangle ) / \sqrt{2})^{\otimes N}$, $| \Phi_1 \rangle= | 0 \rangle^{\otimes N}$, and $c_a = 1/\sqrt{3}$. The environment Hamiltonian is again \eqref{eq:HE-ex}. Thus, the energy densities are redundant, $ 0 = \epsilon_0 = \epsilon_2 \ne \epsilon_1$. We observe in Fig~\ref{fig:exampled3} that the initial redundancy plateau with value $\ln(3)$ gradually crossovers to a partial redundancy plateau with value predicted by \eqref{eq:plateau-partial}, $s = \ln (3)/ 3 + 2 \ln (3/2) /3$. Due to the slow approach to the plateau and the ``encoding jump'' at $N / 2$, the behavior of $I(\mathcal{F}, \mathcal{S})$ is quite complex, and would have been hard to understand without the above analytical insights.

	\begin{figure}
		\centering
		\includegraphics[width =.9 \columnwidth]{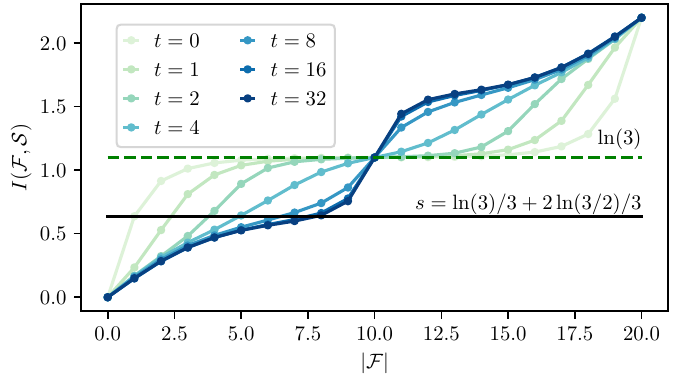}
		\caption{System-fraction mutual information with partial energy degeneracy. At late times, a barely visible partial redundancy plateau at $s = \ln (3)/ 3 + 2 \ln (3/2) /3 $ is followed by an encoding jump at $|\mathcal{F}| = |\mathcal{E}| / 2$, where $|\mathcal{E}| = 20.$}\label{fig:exampled3}
	\end{figure}

\end{document}